\documentclass[runningheads]{llncs}
\usepackage{amsmath}
\usepackage[T1]{fontenc}
\usepackage{graphicx}

\usepackage{hyperref}       
\usepackage{color}

\usepackage[utf8]{inputenc} 
\usepackage{booktabs}       
\usepackage{amsfonts}       
\usepackage{microtype}      
\usepackage{physics}

\usepackage{comment}

\begin{document}

\title{Diverse Neural Sequences in QIF Networks: An Analytically Tractable Framework for Synfire Chains and Hippocampal Replay}
\titlerunning{Diverse Neural Sequences in QIF Networks}
%
\author{Genki~Shimizu\inst{1,2}\orcidID{0000-0002-1393-8060} \and 
Taro~Toyoizumi\inst{1,2}\orcidID{0000-0001-5444-8829}}

\authorrunning{G. Shimizu and T. Toyoizumi}
%
\institute{RIKEN Center for Brain Science\\
2-1 Hirosawa, Wako, Saitama, Japan  \\
\email{\{genki.shimizu@a.riken.jp, taro.toyoizumi@riken.jp\}}
\and
The University of Tokyo \\
7-3-1 Hongo, Bunkyo-city, Tokyo, Japan}

\maketitle              

\begin{abstract}
Sequential neural activity is fundamental to cognition, yet how diverse sequences are recalled under biological constraints remains a key question. Existing models often struggle to balance biophysical realism and analytical tractability.
We address this problem by proposing a parsimonious network of Quadratic Integrate-and-Fire (QIF) neurons with sequences embedded via a temporally asymmetric Hebbian (TAH) rule.
Our findings demonstrate that this single framework robustly reproduces a spectrum of sequential activities, including persistent synfire-like chains and transient, hippocampal replay-like bursts exhibiting intra-ripple frequency accommodation (IFA), all achieved without requiring specialized delay or adaptation mechanisms.
Crucially, we derive exact low-dimensional firing-rate equations (FREs) that provide mechanistic insight, elucidating the bifurcation structure governing these distinct dynamical regimes and explaining their stability.
The model also exhibits strong robustness to synaptic heterogeneity and memory pattern overlap.
These results establish QIF networks with TAH connectivity as an analytically tractable and biologically plausible platform for investigating the emergence, stability, and diversity of sequential neural activity in the brain.
\end{abstract}

\section{Introduction}
\label{sec:intro}
Sequential patterns of spiking activity are ubiquitous in the brain and underpin cognitive functions including decision-making \cite{harvey2012Choicespecific}, motor control \cite{wang2018Flexible}, interval timing \cite{pastalkova2008Internally}, and spatial navigation \cite{lee2002Memory}. They manifest across regions and species, from the millisecond‑precision bursts of HVC neurons that knit together birdsong syllables \cite{hahnloser2002Ultrasparse,fee2004Neural} to the highly compressed hippocampal \emph{replay} of place‑cell trajectories during sharp‑wave ripples (SWRs) \cite{buzsaki2015Hippocampal,olafsdottir2018Role}.

Beyond the sequential ordering, sharp‑wave ripples also display a brief but systematic slowing of their $\sim 200\, \mathrm{Hz}$ oscillation, a phenomenon termed \emph{intra‑ripple frequency accommodation} (IFA). Within a single 50–100ms event the instantaneous frequency typically falls by 30-40\%, and perturbations of perisomatic inhibition strongly modulate this slope \cite{ylinen1995Sharp,nguyen2009Characterizing,sullivan2011Relationships}. 
Although less widely known than replay itself, IFA offers an accessible fingerprint of the circuitry that paces ripples and sets the effective temporal compression of reactivated sequences.

Explaining how rich and precisely timed sequences such as these emerge and remain robust is a long‑standing challenge. 
Existing computational frameworks face trade-offs: 
classical attractor networks often require biologically questionable timescales or adaptation mechanisms \cite{sompolinsky1986Temporal,suri2002Spike}; 
conventional models of synfire chains primarily emphasize precise, stereotypical propagation of spikes over reproducing diverse temporal patterns in neural sequences, and can be overly sensitive to noise and parameter tuning \cite{diesmann1999Stable,long2010Support}; 
standard rate-based models may overlook critical spiking dynamics like transient population synchrony \cite{schaffer2013ComplexValued,gillett2020Characteristics}; 
and detailed biophysical simulations trade analytical tractability for realism.

To address this challenge, we propose a parsimonious framework based on networks of Quadratic Integrate-and-Fire (QIF) neurons, which capture Type~I excitability and are amenable to exact low-dimensional reduction via the Ott-Antonsen ansatz \cite{montbrio2015Macroscopic}. 
We demonstrate that these QIF networks, when endowed with a temporally asymmetric Hebbian (TAH) rule, can robustly generate a spectrum of sequential activities. 
Specifically, our model reproduces both sustained, precisely timed synfire-like chains and transient, hippocampal replay-like sequences that exhibit IFA, all without requiring specialized synaptic delay or slow adaptation currents. 
Furthermore, the derived firing-rate equations (FREs) provide a powerful analytical handle, elucidating the bifurcation structure that governs these distinct dynamical regimes and the mechanisms ensuring their stability. We show that these capabilities are preserved even under significant synaptic heterogeneity and memory pattern overlap. This work positions QIF networks as a unifying and analytically tractable platform for exploring the emergence and characteristics of biologically plausible neural sequences.

\section{Model}
\label{sec:model}

\subsection{Quadratic Integrate-and-Fire neurons}
We base our model on the QIF neuron, a minimal spiking unit capturing Type I excitability near threshold \cite{latham2000Intrinsic,hansel2001Existence}. Each neuron’s membrane potential evolves as:
\begin{equation}
  \dot V = V^2 + \eta + I(t),
  \label{eq:QIF}
\end{equation}
where $\eta$ is its intrinsic excitability and $I(t)$ denotes external input. A spike is emitted as $V\to+\infty$, with an instantaneous reset to $V=-\infty$.

\subsection{Network architecture and learning}
\label{subsec:network}
We store $P$ binary patterns $\{\xi_i^\mu\}$ ($\mu=1,\dots,P$) in synaptic weights via (conventional) Hebbian rule and temporally asymmetric Hebbian (TAH) rule:
\begin{subequations}
\label{eq:weights}
\begin{gather}
  W_{ij}^{\rm Hebb} \propto \sum_{\mu=1}^P \xi_i^\mu \xi_j^\mu, ~~~
  W_{ij}^{\rm TAH}  \propto \sum_{\mu=1}^{P} \xi_i^{\mu+1}\xi_j^\mu, ~~~
  W_{ij}^{\rm inh} \propto  \frac{1}{P} \left(\sum_{\mu=1}^P \xi_i^\mu\right)\left( \sum_{\mu'=1}^P \xi_j^{\mu '}\right)
  \label{eq:weight_component}
  \\
  W_{ij} = W_{ij}^{\rm Hebb} + W_{ij}^{\rm TAH} - W_{ij}^{\rm inh}.
\end{gather}
\end{subequations}
Here, $W^{\rm Hebb}$ creates attractors for each pattern, and $W^{\rm TAH}$ enforces progression from pattern $\mu$ to $\mu+1$.
We also included global inhibitory feedback $W_{ij}^{\rm inh}$ for stability and the ease of analysis.

The dynamics of each neuron in the network are determined by 
\begin{equation}
  \dot V_i = V_i^2 + \eta_i + \sum_{j,f} W_{ij}\,\delta(t - t_j^f - \epsilon) + I_{\rm ext}(t),
  \label{eq:network_dynamics}
\end{equation}
where $t_j^f$ denotes the $f$‐th past spike of neuron $j$, $\epsilon$ represents synaptic conductance delay, and $I_{\rm ext}$ is an external drive.
We set $\epsilon \to 0$ for theoretical analysis and $\epsilon = 10^{-8}$ for numerical simulation.

\subsection{Reduction to firing-rate equations}
For ease of analysis, we consider the sparse-memory limit, 
where overlap between memory patterns will be ignorable so that each neuron participates in at most one memory.
Here, the network decomposes into $P$ disjoint assemblies (indexed by $k$), and synaptic weights induced by memory patterns \eqref{eq:weight_component} become connections between populations.
Following \cite{montbrio2015Macroscopic}, we further apply the Ott–Antonsen reduction in the thermodynamic limit, which yield exact firing‐rate equations (FREs) 
for each population’s firing rate $r_k$ and median voltage $v_k$.\footnote{For the detail of derivation, see Appendix \ref{appendix:FREderivation}.}
\begin{subequations}
\label{eq:FRE}
\begin{align}
  \dot r_k &= \frac{\Delta}{\pi} + 2 r_k v_k,  \label{eq:FRE_r}\\
  \dot v_k &= v_k^2 + \bar\eta + J_1(r_{k-1}-\bar r) + J_2\,r_k - (\pi r_k)^2
  \label{eq:FRE_v}
\end{align}
\end{subequations}
where $J_1$ and $J_2$ are the sequential and recurrent coupling strengths, respectively; $\bar r=\frac1K\sum_k r_k$ is the mean firing rate across all populations; and $\bar\eta$ and $\Delta$ denote the median and half-width of the (Lorentzian) distribution $\eta_i$.
These FREs provide direct analytical insight into the dynamics of sequential activity in the network.

\section{Results}
\label{sec:results}

Although with minimal structure, the QIF network with TAH rule generates robust sequential activity under varied conditions, echoing spatiotemporal complexity of neural sequences in the brain. 
Furthermore, the reduction to FREs provides theoretical insights under the mechanism of sequential activities.
In the following section, we first give linear-stability analysis of the network dynamics.
We then highlight three key phenomena in our model: stable synfire-chain propagation, replay-like transient sequences, and robustness to heterogeneity and pattern overlap.

\subsection{Fixed points and stability of the network}
\label{subsec:stability}
Since feedforward coupling $J_1$ preserves the spatially homogeneous manifold, any equilibrium of the single‐population FREs
\begin{subequations}\label{eq:FRE_single}
\begin{align}
  \dot r &= \tfrac{\Delta}{\pi} + 2rv,\\
  \dot v &= v^2 + \bar\eta + J_2r - (\pi r)^2
\end{align}
\end{subequations}
remains an equilibrium of the full network.  System~\eqref{eq:FRE_single} is bistable\cite{montbrio2015Macroscopic}, with a low‐rate node $(r^*_{\text{node}},v^*_{\text{node}})$ and a high‐rate focus $(r^*_{\text{fc}},v^*_{\text{fc}})$ (Fig.~\ref{fig:single-FRE}A).  Linearizing the $2P$‐dimensional FREs around a homogeneous fixed point yields $P$ independent $2\times2$ blocks with eigenvalues
\begin{equation}\label{eq:eigenvalue}
  \lambda_{k}^{\pm} =
  \begin{cases}
    2v^{\ast} \;\pm\; \sqrt{2 r^{\ast}\bigl(J_{2}-2\pi^{2}r^{\ast}\bigr)}, & k=0,\\[6pt]
    2v^{\ast} \;\pm\; \sqrt{2 r^{\ast}\bigl(J_{2}-2\pi^{2}r^{\ast}+J_{1}\,\zeta_{P}^{k}\bigr)}, & k=1,\dots,P-1.
  \end{cases}
\end{equation}
where $\xi_{P}=e^{2\pi i/P}$. $k=0$ recovers the single‐population mode and $k\ge1$ captures sequential perturbations via $J_1\,\zeta_P^k$ (see Appendix~\ref{appendix:eigenvalues}).  In the $P\to\infty$ limit, the leading growth rate becomes
\begin{align}
\label{eq:growthrate}
2v^* + \sqrt{2r^*}\,
\begin{cases}
\sqrt{J_1 + J_2 - 2\pi^2r^*}, & J_1 + 2J_2 \ge 4\pi^2r^*,\\
\frac{J_1}{2}\,|J_2 - 2\pi^2r^*|^{-1/2}, & \text{otherwise}.
\end{cases}
\end{align}
Thus, increasing $J_1$ first destabilizes the high‐rate focus while the low‐rate node remains stable (Fig.~\ref{fig:single-FRE}B,C), defining the regime for transient replay.
This bifurcation structure predicts that a slow modulatory input can drive transitions between the focus and nodes, enabling replay-like damped oscillation during the depolarized up-state and suppressing activity in the hyperpolarized down-state.

\begin{figure}[t]
    \centering
    \includegraphics[width=\linewidth]{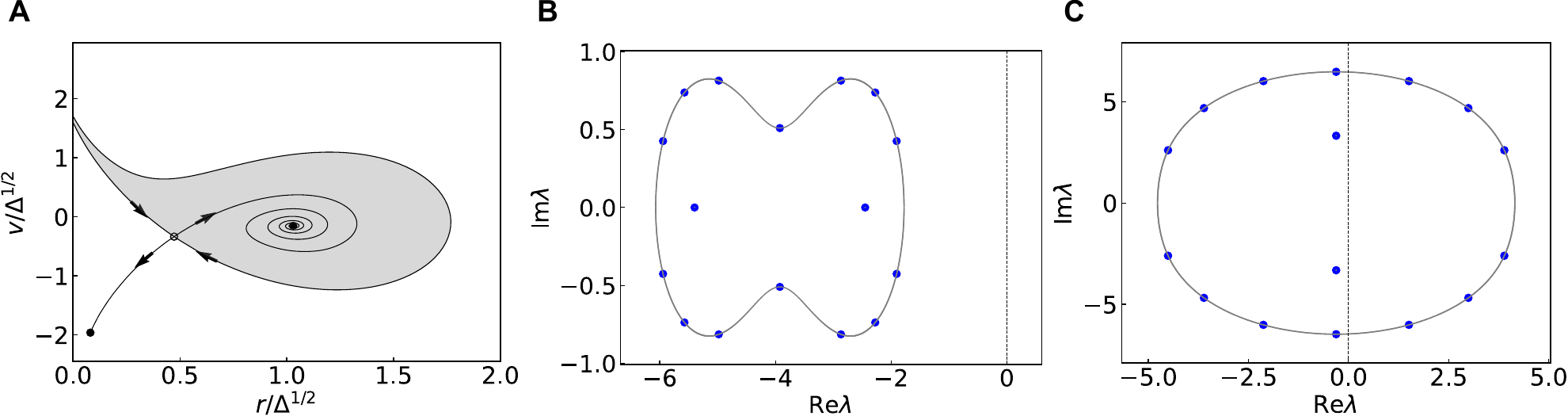}
    \caption{\textbf{Phase-plane structure and linear stability of the firing-rate equations.}
    \textbf{A}~ Single-population FREs~\eqref{eq:FRE_single} are bistable, with a low-rate node $(r^{*}_{\mathrm{node}},v^{*}_{\mathrm{node}})$ (bottom left) and a high-rate focus $(r^{*}_{\mathrm{fc}},v^{*}_{\mathrm{fc}})$ (upper right).  
    \textbf{B}~ Eigenvalue spectrum of the full Jacobian at $(r_k, v_k) = (r^{*}_{\mathrm{fc}},v^{*}_{\mathrm{fc}})$. Solid dots show the discrete spectrum for a finite network with $P=8$ populations; the central dot corresponds to the spatially homogeneous mode ($k=0$), which matches the single-population eigenvalue. The grey dotted curve depicts the continuum envelope obtained as $P\rightarrow\infty$.  
    \textbf{C}~ Same conventions as in (b) but for $(r^{*}_{\mathrm{node}},v^{*}_{\mathrm{node}})$, illustrating that the node remains stable while the focus becomes unstable when $J_{1}\approx J_{2}$.  
    Parameters shared across panels: $J_{1}=15$, $J_{2}=15$, $\bar{\eta}=-5$, $\Delta=1$.}
    \label{fig:single-FRE}
\end{figure}

\subsection{Stable synfire–like propagation}
\label{subsec:synfire}
Starting from the low-rate equilibrium $(r^*_{\mathrm{node}}, v^*_{\mathrm{node}})$, a brief external pulse excites the first population.  
Feedforward connections then recruit each subsequent population in turn, creating a self-sustaining feed-forward chain (Fig.~\ref{fig:synfire}A).  
The FREs predict both burst timing and amplitude with high accuracy, matching full spiking simulations (Fig.~\ref{fig:synfire}B,C).

\begin{figure}[t]
    \centering
    \includegraphics[width=\linewidth]{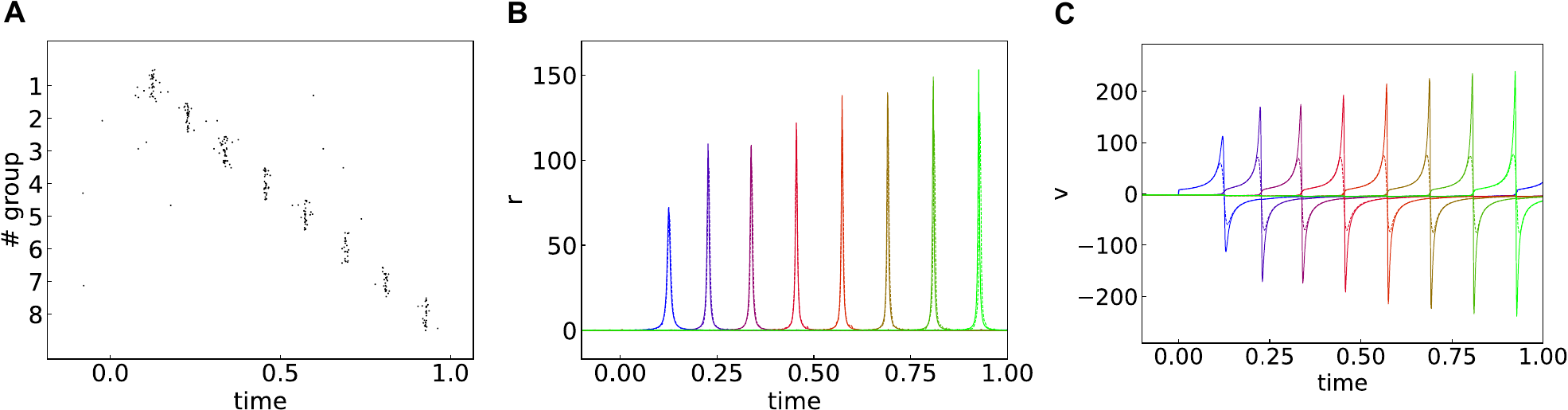}
    \caption{\textbf{Stable synfire propagation in a QIF network with TAH connectivity.}  
    \textbf{A}~ Spike raster: synchronized bursts travel successively among populations.  
    \textbf{B}~ Population firing rate; solid lines are FRE predictions, dashed lines are full simulations.  
    \textbf{C}~ Median membrane potential, same conventions as in B.  
    Parameters: $J_1=15$, $J_2=15$, $\bar{\eta}=-5$, $\Delta=1$, $P=8$.}    
    \label{fig:synfire}
\end{figure}

To clarify why propagation remains stable, we investigated the response of a single population \eqref{eq:FRE_single} to a Gaussian pulse  
\[
I_{\text{ext}}(t)=
A_{\text{in}}\,(2\pi\sigma_{\text{in}}^{2})^{-1/2}
\exp\!\bigl[-t^{2}/(2\sigma_{\text{in}}^{2})\bigr].
\]
Depending on $(A_{\text{in}},\sigma_{\text{in}})$, the collective dynamics fall into three qualitatively distinct regimes:
\begin{enumerate}\setlength\itemsep{2pt}
\item[\textbf{(i)}] \emph{Synchronized burst.}  
    The population emits a sharp pulse, then relaxes back to the low-rate node $(r^{*}_{\mathrm{node}},v^{*}_{\mathrm{node}})$.  
    We characterise the burst by its peak area $A_{\text{out}}$ and width $\sigma_{\text{out}}$.
\item[\textbf{(ii)}] \emph{Weak, quickly vanishing response.}  
    A subthreshold deflection that also returns to the node, again summarised by $(A_{\text{out}},\sigma_{\text{out}})$ with much smaller amplitude.
\item[\textbf{(iii)}] \emph{Transition to sustained activity.}  
    For the input with intermediate strength and/or precision, the trajectory leaves the node’s basin and converges to the high-rate focus $(r^{*}_{\mathrm{fc}},v^{*}_{\mathrm{fc}})$.
\end{enumerate}

For regimes (i) and (ii) the mapping
\[
(A_{\text{in}},\sigma_{\text{in}})\;\longmapsto\;
(A_{\text{out}},\sigma_{\text{out}})
\]
defines a well-behaved vector field on the $(A,\sigma)$ state space (Fig.~\ref{fig:response})
State-space analysis shows that this map has a stable fixed point representing a stereotyped, high-precision burst.  
Because every population in the feed-forward chain receives input that lies inside the basin of this fixed point, each group reproduces essentially the same output, preventing both runaway excitation and decay.  
This mechanism explains the remarkably uniform spike volleys seen in the full network simulation.

\begin{figure}[t]
    \centering
    \includegraphics[width=\linewidth]{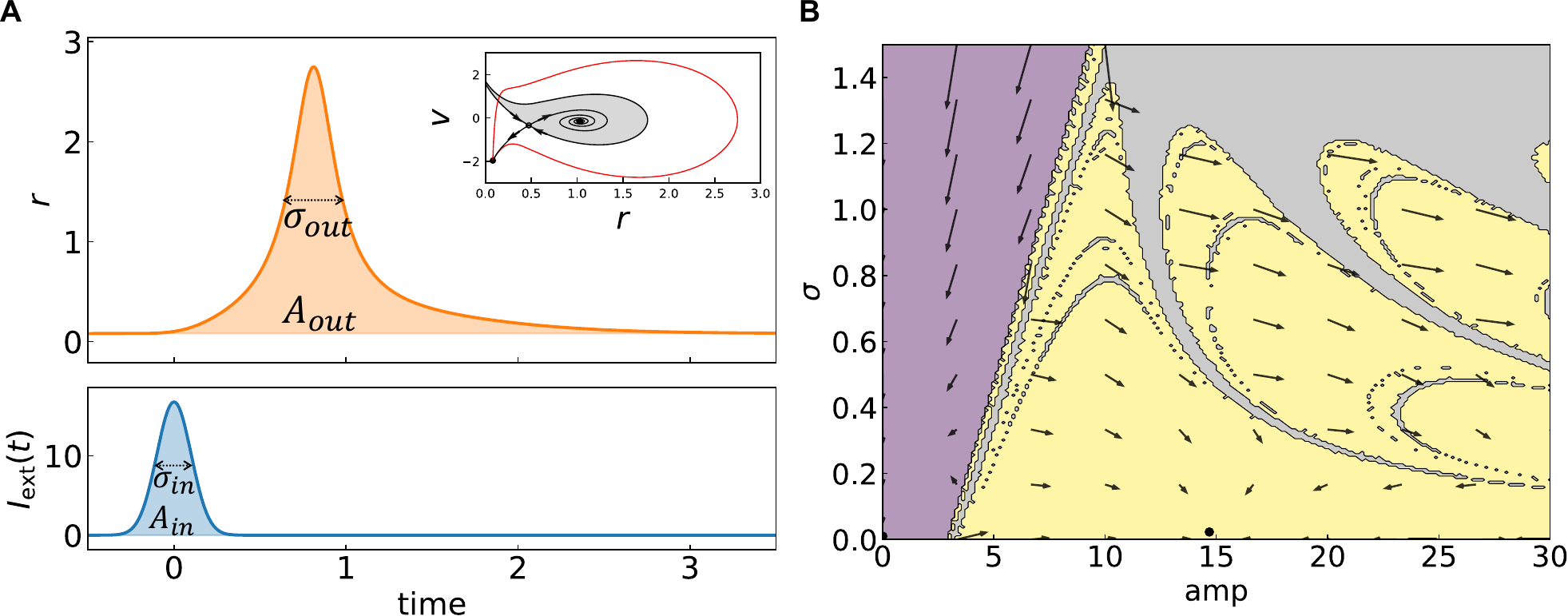}
    \caption{\textbf{Single-population response explains stability of synfire propagation.}  
    \textbf{A}~ Response to a Gaussian input ($A_{\text{in}},\sigma_{\text{in}}$). Before stimulation, the population rested at the low-rate node.  
    \textbf{B}~ State-space analysis of the response. Two fixed points exist: the origin (quiescent state) and a bursting attractor (black dot). Purple and yellow shading indicate their basins of attraction; grey trajectories lead to the high-rate focus. Arrows show the vector field.  
    Parameters: $J_2=15$, $\bar{\eta}=-5$, $\Delta=1$.}
    \label{fig:response}
\end{figure}

\subsection{Replay-like transient sequences}
\label{subsec:replay}

We next explore a dynamical regime in which the network shows transient oscillations around the high-rate focus, giving rise to brief episodes of sequential activity that resemble hippocampal replay during sharp-wave ripples (SWRs). 
We therefore set the sequential coupling weak, $J_{1}\ll J_{2}$, so that the focus $(r^{*}_{\mathrm{fc}},v^{*}_{\mathrm{fc}})$ remains linearly stable.  
A slow, subthreshold modulation,
\[
I_{\text{mod}}(t)=I_{0}\bigl[1-\cos(2\pi f_{\text{mod}} t)\bigr],
\]
mimics the depolarising phase of cortical up/down states.  
As the drive increases, the trajectory leaves the low-rate node, executes a spiral excursion near the focus, and returns to the node as the input wanes.

A representative event is shown in Fig.~\ref{fig:ripple-replay}A.  
During each “Up” excursion, the population firing rate oscillates rapidly.  
Within these oscillations, activity sweeps successively through the populations, replaying the stored sequence several times before dying out (Fig.~\ref{fig:ripple-replay}B).  
The low-dimensional FREs predict both the latency to onset and the number of within-burst cycles with high accuracy.

\begin{figure}[t]
    \centering
    \includegraphics[width=\linewidth]{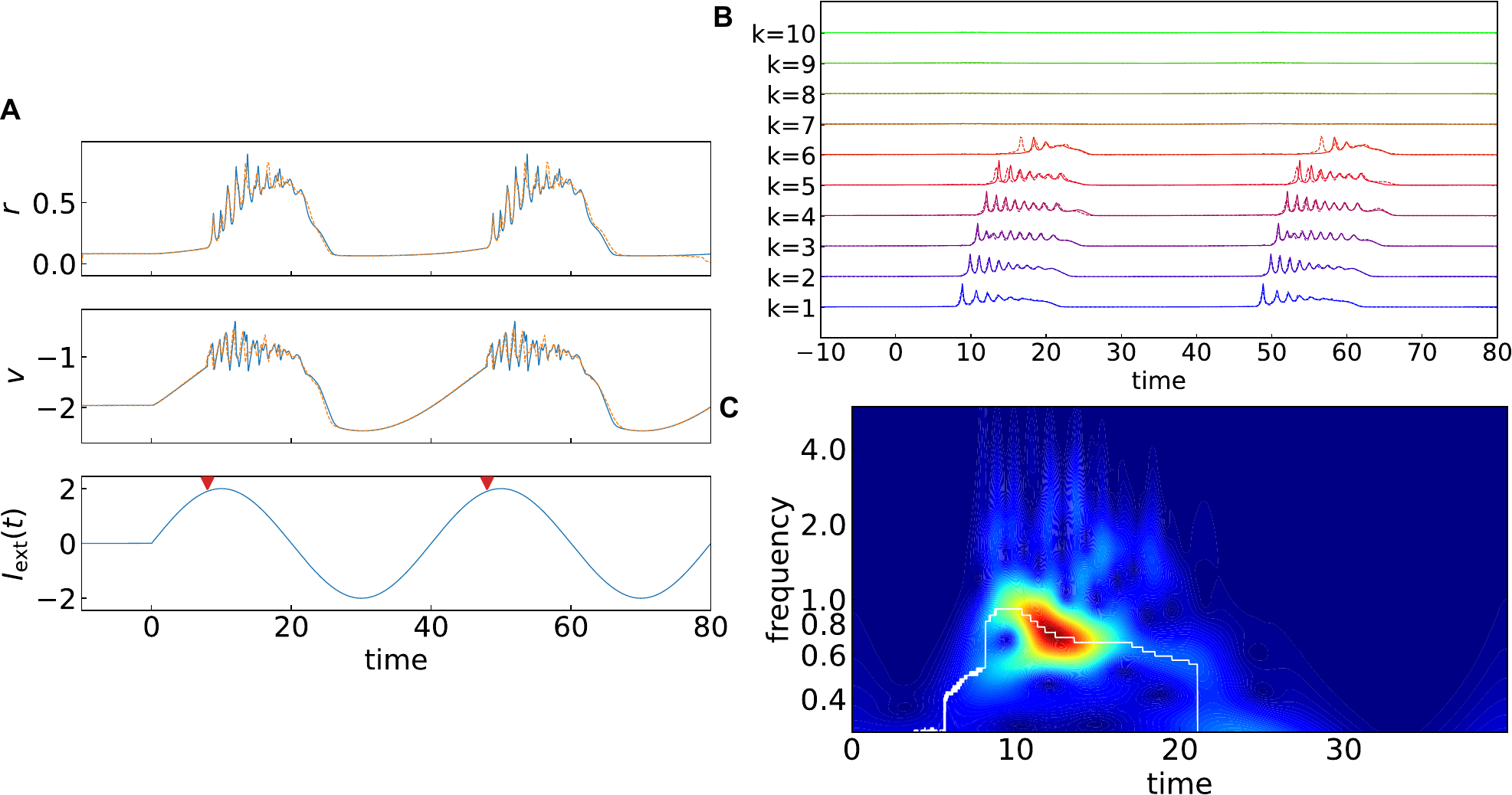}
    \caption{\textbf{Transient replay and ripple oscillations under slow modulatory drive.}
    \textbf{A}~ Representative SWR event. Solid lines: FRE prediction; dashed lines: full simulation.  
    \textbf{B}~ Within each burst, sequential activity propagates several times through the populations before fading (same conventions as in A).  
    \textbf{C}~ Wavelet spectrogram; the white curve traces the instantaneous peak frequency, revealing IFA.
    Parameters: $J_{1}=1.8$, $J_{2}=15$, $\bar{\eta}=-5$, $\Delta=1$, $P=10$.}
    \label{fig:ripple-replay}
\end{figure}

A striking feature of the spectrogram (Fig.~\ref{fig:ripple-replay}C) is the systematic \emph{down-chirp} of ripple frequency, mirroring the intra-ripple frequency accommodation (IFA) observed \emph{in vivo} and \emph{in vitro} across species and brain states\cite{sullivan2011Relationships,ponomarenko2004Multiple,nguyen2009Characterizing}.

These results show that weak sequential coupling combined with slow modulatory drive provides a minimal, analytically tractable mechanism for replay-like bursts with realistic IFA, extending the stable synfire-chain propagation analyzed in Section~\ref{subsec:synfire}.

\subsection{Robustness to heterogeneity and pattern overlap}
\label{subsec:robustness}

Real cortical networks are neither homogeneous nor store perfectly orthogonal memories.  
We therefore tested whether our model’s sequential recall survives two realistic perturbations.

\paragraph{Synaptic heterogeneity.}
Each deterministic weight \(W_{ij}^{\text{det}}\) in Eq.~\eqref{eq:weights} was multiplied by a lognormal factor
\(\exp(\sigma_{\text{syn}}\zeta_{ij})\) with \(\zeta_{ij}\sim\mathcal N(0,1)\).  
Even for large spreads (\(\sigma_{\text{syn}}=1\!-\!2\)), spike volleys still crossed all \(P\) populations, though they became slightly broader and slower (Fig.~\ref{fig:noise}A).

\paragraph{Memory overlap.}
Binary patterns were drawn with sparsity \(f\), giving an expected overlap of \(f^{2}\) between any two memories.  
For \(f=0.01\)–\(0.1\), the correct sequence re-emerged without mis-routing (Fig.~\ref{fig:noise}B), again only widening and slowing the bursts.

Both perturbations keep the input pair \((A_{\text{in}},\sigma_{\text{in}})\) inside the basin of attraction described in Section~\ref{subsec:synfire}, thereby preserving stereotyped feed-forward propagation.

\begin{figure}[t]
    \centering
    \includegraphics[width=\linewidth]{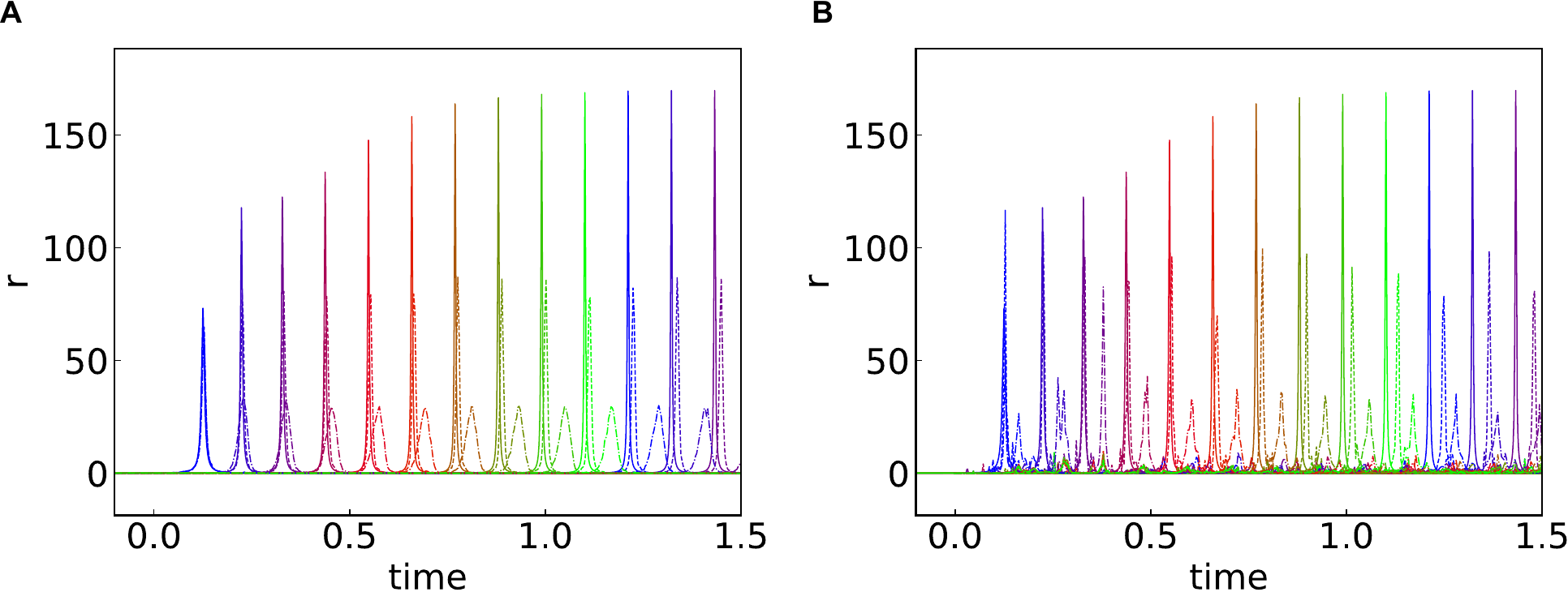}
    \caption{\textbf{Robustness of sequential propagation.}
    \textbf{A} Lognormal synaptic noise (\(\sigma_{\text{syn}}=1.0\), dotted; \(2.0\), dashed) broadens and slows volleys but does not halt them.  
    \textbf{B} Pattern overlap (sparsity \(f=0.01\), dotted; \(f=0.1\), dashed) likewise preserves the sequence.  
    Solid lines: unperturbed FRE prediction.}
    \label{fig:noise}
\end{figure}

\section{Discussion}
\label{sec:discussion}
We have shown that a compact QIF network equipped with a temporally asymmetric Hebbian (TAH) rule can reproduce a wide range of sequential dynamics while remaining fully amenable to analysis.  Leveraging an exact low-dimensional firing-rate reduction alongside large-scale spiking simulations, we identified four main findings:

\begin{itemize}
  \item \textbf{Bifurcation structure:} The reduced system possesses two coexisting fixed points—a low-rate node and a high-rate focus—whose stability is controlled by the sequential coupling $J_1$.
  \item \textbf{Stable synfire propagation:} When $J_1$ is sufficiently strong, a brief input launches self-sustaining spike volleys that move cleanly from one population to the next; a single-population input–output map explains their constant gain and width.
  \item \textbf{Replay-like bursts with IFA:} For $J_1\ll J_2$, a slow modulatory drive triggers transient spiral excursions around the focus, producing replay-like bursts whose ripple frequency decays within each event—recapitulating intra-ripple frequency accommodation (IFA) seen in vivo.
  \item \textbf{Robustness:}  Significant synaptic heterogeneity and pattern overlap broaden and slow the volleys but do not extinguish them, underscoring the model’s tolerance to biological variability.
\end{itemize}

Together, these results position QIF networks with TAH connectivity as a minimal yet predictive framework that unifies stable synfire chains and hippocampal replay within a single set of circuit principles.

\paragraph{Comparison to previous frameworks.}
Classical attractor-sequence models (e.g., \cite{sompolinsky1986Temporal,suri2002Spike}) embed sequences via asymmetric weights or slow adaptation but often require non-physiological timescales or strong adaptation to exit each attractor and avoid unwanted synchrony. 
Synfire chains (e.g., \cite{diesmann1999Stable,long2010Support}) achieve precise spike volleys, but timing jitter and parameter drifts can easily trigger runaway activity.
Rate-based sequential models (e.g., \cite{gillett2020Characteristics}) are powerful and offer a degree of analytical simplicity, yet they often struggle to capture precise synchronization and transient dynamics observed both in the brain and spiking models, such as damped oscillations \cite{schaffer2013ComplexValued}.

Our QIF network framework offers a compelling alternative. 
Its key advantage lies in the combination of biophysically plausible QIF neurons with an exact low-dimensional reduction to FREs\cite{montbrio2015Macroscopic}. 
Despite its analytical tractability, the model robustly reproduces diverse sequential patterns, including the transient spike dynamics which classical binary neurons and even standard rate models typically fail to reproduce \cite{montbrio2015Macroscopic,schaffer2013ComplexValued}. 
This allows our model to bridge the gap between abstract rate models and complex, detailed simulations.

\paragraph{Mechanistic insights.}
The simplicity of the FRE framework provides several crucial mechanistic insights into sequence generation. 

For example, while synfire-chain propagation constitutes a limit-cycle solution, whose stability generally transcends simple linear analysis of fixed points, the specific structure of our model (sparse activity and restricted conductivities) allows a tractable state-space analysis of a neural response. This analysis reveals how the invariance of response gain and latency underpins stable propagation, preventing both runaway excitation and decay. Besides, eigenvalue analysis of the FREs around the high-rate fixed point elucidates the conditions under which transient sequences can exist stably. 

Thus, direct linkage between network parameters and macroscopic activity described by FREs provides a clear and predictive understanding of the sequential activities.

\paragraph{Biological implications.}

Our QIF network not only generates ripple-frequency oscillations accompanied by transient replay, but also \emph{spontaneously} reproduces the hallmark phenomenon of \emph{intra-ripple frequency accommodation} (IFA)—the systematic deceleration of ripple frequency within a single event.  
IFA has been observed across species and preparations\cite{sullivan2011Relationships,ponomarenko2004Multiple,nguyen2009Characterizing} and is thought to tailor the temporal compression of replay so as to maximise spike-timing–dependent plasticity and facilitate memory consolidation\cite{schieferstein2024Intraripple}.  
Because mechanistic and functional explanations remain scarce, the ability of our minimal QIF framework to capture this nuanced dynamical signature suggests that it preserves the essential circuit ingredients required for \textit{in vivo} sequence generation.

The model also makes a clear testable prediction: introducing biological variability slows and broadens spike volleys (Fig.~\ref{fig:noise}).  Quantifying how replay speed and temporal precision covary with cellular- and synaptic-level noise therefore provides a straightforward experimental assay of our theoretical proposal.

\paragraph{Limitations and future work.}

The present model stores a fixed repertoire of sequences learned offline.  
Future work should endow the network with \emph{online} plasticity that can insert, update, or delete sequences without interference, and with gating mechanisms that enable context-dependent switching.  
Also, extending the framework to variable replay speeds, branched trajectories, and richer biophysical detail (short-term plasticity, intrinsic adaptation) will further test its flexibility.  

\section{Conclusion}
\label{sec:conclusion}
We proposed a parsimonious QIF network with temporally asymmetric Hebbian connectivity that stores and recalls neural sequences.  
Exact low-dimensional firing-rate equations reveal the bifurcation landscape that supports both stable synfire propagation and replay-like bursts with intra-ripple frequency accommodation, all without delay lines or strong adaptation.  
The mechanism remains robust under synaptic noise and overlapping memories.  
Thus, QIF networks with TAH connectivity offer an analytically transparent platform for studying the generation, stability, and modulation of sequential activity in spiking circuits.

\begin{credits}
\subsubsection{\ackname} 
The author gratefully acknowledges Kensuke Yoshida and all members of the Toyoizumi Lab for insightful advice and stimulating discussions. 
This work was funded by JSPS KAKENHI Grant Number 23KJ0666, JST CREST JPMJCR23N2, RIKEN Center for Brain Science, and RIKEN TRIP initiate (RIKEN Quantum). 
Heartfelt gratitude to my wife and newborn child for their unwavering support and inspiration.

This preprint has not undergone peer review or any post-submission improvements or corrections.
The Version of Record of this contribution will be published in the ICONIP 2025 proceedings.
The DOI will be added when available.

\subsubsection{\discintname}
The authors have no competing interests to declare that are
relevant to the content of this article. 

\subsubsection{Code availability} 
All simulation and analysis scripts used in this study are openly
available at 
\url{https://github.com/genkinanodesu/QIF-sequence}.
\end{credits}

\bibliographystyle{splncs04}
\bibliography{qif-sequence}

\appendix
\numberwithin{equation}{section}

\section{Derivation of FREs}
\label{appendix:FREderivation}

In this appendix we sketch how the low-dimensional firing-rate equations (FREs),
used throughout the main text, arise from the microscopic dynamics of the
sequential-recall QIF network introduced in
Section~\ref{subsec:network}.  The derivation closely follows
Montbrió \textit{et al.}~\cite{montbrio2015Macroscopic}, who treated the
single-population case.

Starting point is the membrane-potential dynamics of neuron $i$,
\begin{align}
  \dot V_{i}
  &= V_{i}^{2} \;+\;\eta_{i}\;+\;J\,s_{k}(t)\;+\;I(t),
  \label{eq:micro}\\
  s_k(t)
  &= \frac1N \sum_{i=1}^{N}\sum_{f}\delta\!\bigl(t-t^{f}_{i}\bigr),
\end{align}
where $s_k(t)$ is the population spike train, $J$ the coupling strength,
and $I(t)$ an external drive.  Under three key assumptions, Montbrió
\textit{et al.} showed that the collective behaviour can be described by
two macroscopic variables—the firing rate $r(t)$ and the mean membrane
potential $v(t)$—which obey Eq.~\eqref{eq:FRE_single}:

\begin{enumerate}
  \item \textbf{Thermodynamic limit:} the number of neurons
        $N\to\infty$.
  \item \textbf{Lorentzian ansatz (LA):} for each intrinsic current
        $\eta$, the conditional potential density
        $\rho(V\!\mid\!\eta,t)$ is Lorentzian in $V$.
  \item \textbf{Cauchy-distributed excitabilities:}
        the intrinsic currents $\eta_{i}$ are i.i.d.\ samples from a
        Cauchy (Lorentzian) density $g(\eta)$.
\end{enumerate}

\paragraph{Extension to Hebbian/TAH memory.}
To incorporate our Hebbian/temporally asymmetric Hebbian (TAH)
connectivity, we take the sparse-memory limit in which overlaps between
stored patterns are negligible.  Neurons can then be partitioned into
$P$ disjoint populations,
\[
  \text{group }k \;:=\;
  \bigl\{\,\text{neuron }i \;\bigm|\; \xi_{i}^{k}=1 \bigr\},
\]
and neurons not belonging to any pattern are ignored.  Assuming identical
sparsity across patterns, each group contains $N$ neurons. In this setting,
\begin{itemize}
  \item the Hebbian term $W^{\mathrm{Hebb}}_{ij}$ acts as a recurrent
        connection within group $k$,
  \item the TAH term $W^{\mathrm{TAH}}_{ij}$ provides feed-forward coupling from $k$ to $k+1$,
  \item global inhibitory feedback $W^{\mathrm{inh}}_{ij}$ yields uniform all-to-all inhibition.
\end{itemize}

The resulting FREs for population~$k$ are
\begin{subequations}
\begin{align}
  \dot r_k &= \frac{\Delta}{\pi} + 2\,r_k\,v_k, \\[2pt]
  \dot v_k &= v_k^{2} + \bar\eta
             + J_1 r_{k-1}
             + J_2 r_k
             - J_3 \bar r
             - (\pi r_k)^{2},
\end{align}
\end{subequations}
where
\(
  \bar r = P^{-1}\sum_{l=1}^{P} r_l
\)
is the network-wide mean firing rate.  Setting $J_3 = J_1$ guarantees
that these equations share the same fixed points as the single-population
FREs~\eqref{eq:FRE_single}, simplifying the subsequent stability
analysis.

\section{Explicit formulation of eigenvalues of FREs}
\label{appendix:eigenvalues}
\paragraph{Full Jacobian matrix}
Computing 2nd order derivatives of RHS of (\ref{eq:FRE}), we have 
\begin{subequations}
\label{eq:2nd-derivation}
\begin{align}
    \pdv{\dot{r_k}}{r_i} &= 2\delta _{i,k} v_k, &  \pdv{\dot{r_k}}{v_i} &= 2\delta _{i,k} r_k, \\
    \pdv{\dot{v_k}}{r_i} &= J_1 \delta _{i, k-1} + (J_2 - 2\pi^2 r_k) \delta _{i,k} - \frac{J_1}{P}, & \pdv{\dot{v_k}}{v_i} &= 2\delta _{i,k} v_k.
\end{align}
\end{subequations}
Thus, the Jacobian matrix of (\ref{eq:FRE}) around $(r_k, v_k) = (r^*, v^*)$ (for all $k$) is a block-circulant matrix 
\begin{align}
    C &= \label{eq:jacobian}
    \begin{pmatrix}
        C_1-C_3 & C_2-C_3 & -C_3    & \dots   & -C_3  \\
        -C_3    & C_1-C_3 & C_2-C_3 &         & -C_3 \\
        \vdots  & -C_3    &C_1-C_3  &\ddots   & \vdots \\
        -C_3    &         &\ddots   & \ddots  & C_2-C_3\\
        C_2-C_3 & -C_3    &\dots    &  -C_3   &C_1-C_3
    \end{pmatrix} \\
    &= I_P \otimes (C_1-C_3) + P_P \otimes (C_2 - C_3) - (P_P^2 + P_P^3 + \cdots + P_P^{P-1}) \otimes C_3,
\end{align}
where $I_P$ is the $P\times P$ identity matrix, $P_P$ is the $P\times P$ permutation matrix, and 
\begin{align}
    \label{eq:jacobian_matrices}
    C_1 = 
    \begin{pmatrix}
    2v^* & 2r^* \\ C_2-2\pi^2 r^* & 2v^* 
    \end{pmatrix}
    , \quad C_2 = 
    \begin{pmatrix}
        0 &0 \\ C_1 & 0
    \end{pmatrix}
    , \quad C_3 = 
    \begin{pmatrix}
        0 & 0 \\ C_1/P & 0.
    \end{pmatrix}
\end{align}
\paragraph{Block-diagonalization of the Jacobian}

Let $F_P=(\zeta _P^{jk})_{j,k=0}^{P-1}/\sqrt{P}$ be the $P\times P$ unitary discrete Fourier matrix with primitive root $\zeta_P=e^{2\pi i/P}$.
It is well-known that the permutation matrix $P_P$ and its power can be diagonalized by $F_P$ as
\begin{align}    
  F_P^{\dagger}P_PF_P
  =\operatorname{diag}\!\bigl(\zeta_P^{\,k}\bigr)_{k=0}^{P-1},
  \qquad
  F_P^{\dagger}P_P^{\,m}F_P
  =\operatorname{diag}\!\bigl(\zeta_P^{\,km}\bigr)_{k=0}^{P-1}.
\end{align}

Similarly, we define $Q:=F_P\otimes I_2$ and applying to the block-circulant Jacobian~\eqref{eq:jacobian} yields
\begin{align}
  \Lambda =Q^{\dagger} C Q
  =\operatorname{diag}\!\bigl(B_0,B_1,\dots,B_{P-1}\bigr),
\end{align}
where each block $B_k\in\mathbb{C}^{2\times2}$ is  
\begin{align}
B_k =
  \bigl(C_1-C_3\bigr)
  +
  \zeta_P^k\bigl(C_2-C_3\bigr)
  -
  C_3\sum_{m=2}^{P-1}\zeta_P^{\,km}.
\end{align}

Because $\sum_{m=0}^{P-1}\zeta_P^{km}=0$ for $k\neq0$ and $=P$ for $k=0$,
\begin{align}
\label{eq:B_k}
  \boxed{\;
      B_k=
      \begin{cases}
        C_1 + C_2 - PC_3, & k=0,\\
        C_1+\zeta_P^{\,k}C_2, & k=1,\dots,P-1.
      \end{cases}}
\end{align}
Substituting $C_1$, $C_2$ and $C_3$ into \eqref{eq:B_k} yields

\begin{align}
  B_k=
  \begin{pmatrix}
    2v^* & 2r^*\\
    J_2-2\pi^2r^*+J_1\,(1-\delta_{k0})\,\zeta_P^{\,k} & 2v^*
  \end{pmatrix}.
\end{align}

\paragraph{Eigenvalues}
The characteristic polynomial of $B_k$ is
\[
  \det(\lambda I_2-B_k)
  \;=\;
  (\lambda-2v^*)^2
  -2r^*
    \Bigl(
      J_2-2\pi^2r^*+J_1(1-\delta_{k0})\zeta_P^{\,k}
    \Bigr).
\]

Hence, the eigenvalues of the full Jacobian $C$ are
\begin{align}    
    \lambda_k^{\pm}
    =2v^*
     \;\pm\;
     \sqrt{\,2r^*\,\Bigl(
       J_2-2\pi^2r^*
       +J_1(1-\delta_{k0})\zeta_P^{\,k}
     \Bigr)}
  \qquad k=0,1,\dots,P-1.
  \label{eq:eigen_formula}
\end{align}
For the spatially homogeneous mode $k=0$ the $J_1$ term indeed cancels,
while for $k\neq0$ it appears with a complex phase
$\zeta_P^{\,k}=e^{2\pi i k/P}$.

\paragraph{Stability criterion}
The linear stability of the fixed point $(r^*, v^*)$ can be determined by the largest value of $\operatorname{Re}\lambda_k^{\pm}$. 
\begin{lemma}
Let $a$ be a real number. Let $z \in \mathbb{C}$ vary such that $|z|=1$.
Let $w = (z+a)^{1/2}$, and we have
\begin{align}
\max \operatorname{Re} w &= \begin{cases} \sqrt{a+1} & (a \geq -1/2) \\ \frac{1}{2\sqrt{|a|}} & (a < -1/2) \end{cases} 
\end{align}
\end{lemma}
\begin{proof}
Let $w = u + iv$. Since $|z|=1$, the locus of $w$ is such that $|w^2-a|=1$.
Rewriting in terms of $u$ and $v$ yeilds
$$ (u^2 + v^2)^2 - 2a(u^2 -v^2) + (a^2-1) =0. $$
Furthermore, let $X=u^2$ and $Y=v^2$. We then seek the range of $X \geq 0$ for which the following equation in $Y$:
$$ Y^2 + 2(X+a) Y + X^2-2aX + a^2-1 = 0 $$
has a solution in the range $Y \geq 0$.

After some standard (but somewhat tedious) algebra, one yields 
\begin{align*}
\max X &= \begin{cases} a+1  &(a \geq -1/2) \\ \frac{1}{4|a|} & (a < -1/2),  \end{cases} \\
\min X &= \begin{cases} a-1  &(a \geq 1) \\ 0 & (a < 1). \end{cases}
\end{align*}

Since $\max \operatorname{Re} w = \max u = \max \sqrt{X}$, this concludes the proof.
\end{proof}

Obtaining the leading growth rate \eqref{eq:growthrate} is now straightforward by applying the lemma to the formula of eigenvalues \eqref{eq:eigen_formula} and evaluating $\lim _{P\rightarrow \infty }\max_{k}\,\mathrm{Re}\,\lambda_{k}^{\pm} $.

\end{document}